\documentclass[USenglish]{lipics}

\usepackage{graphicx}
\usepackage{amssymb, amsmath}

\newcommand{\COMMENT}[1]{}

\newtheorem*{theorem*}{Theorem}

\newtheorem{proposition}{Proposition}

\newcommand{\R}{\mathbb R}
\newcommand{\N}{\mathbb N}

\newcommand{\ch}{\text{ch}}
\newcommand{\conv}{\text{ch}}
\newcommand{\lift}{\text{lift}}
\newcommand{\proj}{\text{proj}}

\newcommand{\eps}{\varepsilon}

\begin{document}

\title{Erd\H os-Szekeres and Testing Weak $\eps$-Nets are NP-hard in $3$ dimensions - and what now?}
\titlerunning{}

\author[1]{Christian Knauer}
\author[2]{Daniel Werner\footnote{This research was funded by Deutsche Forschungsgemeinschaft within the Research Training Group (Graduiertenkolleg) ``Methods for Discrete Structures''.}}
\affil[1]{Institut f\"ur Informatik,  Universit\"at Bayreuth, Germany 
\texttt{christian.knauer@uni-bayreuth.de}
}
\affil[2]{Institut f{\"ur} Informatik, Freie Universit{\"a}t Berlin, Germany, 
\texttt{ daniel.werner@fu-berlin.de}}
\authorrunning{}

\date{\today}                                      

\maketitle

\begin{abstract}
The Erd\H os-Szekeres theorem states that, for every $k$, there is a number $n_k$ such that every set of $n_k$ points in general position in the plane contains a subset of $k$ points in convex position. If we ask the same question for subsets whose convex hull does not contain any other point from the set, this is not true: as shown by Horton, there are sets of arbitrary size that do not contain an empty $7$-gon.

These questions have also been studied extensively from a computational point of view, and polynomial time algorithms for finding the largest (empty) convex set have been given for the planar case. In higher dimension, it was not known how to compute such a set efficiently. In this paper, we show that already in dimension $3$ no polynomial time algorithm exists for determining the largest (empty) convex set (unless P=NP), by proving that the corresponding decision problem is NP-hard. This answers a question by Dobkin, Edelsbrunner and Overmars from 1990.

As a corollary, we derive a similar result for the closely related problem of testing weak $\eps$-nets in $\R^3$. Answering a question by Chazelle et al. from 1995, our reduction shows that the problem is co-NP-hard.
\end{abstract}

\section{Preliminaries}

The Erd\H os-Szekeres theorem \cite{ES35} is one of the major theorems from combinatorial geometry and one of the earliest results in geometric Ramsey theory.
\begin{theorem*} \textnormal{(Erd\H os and Szekeres, 1935)} For every $k$ there is a number $n_k$ such that every planar set of $n_k$ points in general position contains $k$ points in convex position.
\end{theorem*}
Exact values of $n_k$ are known only for very few cases and subject to extensive research, also for the higher dimensional cases.

A closely related question is the following: is the theorem still true if we ask for sets whose convex hull is empty, i.e., does not contain any other point from the set? That this is not the case was shown by Horton \cite{Ho83}: in the plane there are arbitrary large sets which do not contain empty $7$-gons. Nicol\'as \cite{Nic07} and Gerken \cite{Ger08} independently solved the long standing open problem whether or not there is always an empty $6$-gon.

Both these questions generalize to dimension larger than $2$ in the obvious way, and clearly the numbers $n_k$ do not increase when the dimension gets larger (proof: project to $\R^2$). See the surveys by B\'ar\'any and K\'arolyi \cite{BK01} or Morris and Soltan \cite{MS00} for further references and (more or less) recent progress on the subject.

The corresponding computational problems have also received a lot of attention in the past (e.g., \cite{AR85}, \cite{CK80}, \cite{DEO90}, \cite{MRSW95}). Polynomial time algorithms are known for both problems in the plane. The fastest algorithm is given in \cite{DEO90}, and the question is stated whether a polynomial time algorithm for determining the largest empty convex set also exists in $\R^3$.

\subsection{Our results}
In this paper, we will consider the following decision problems:
\begin{definition} \textsc{(Erd\H os-Szekeres)} Let $P$ be a set of points in $\R^d$ and $k \in \N$. Is there a set $Q \subseteq P$ of $k$ points in convex position?
\end{definition}
and
\begin{definition} \textsc{(Largest-Empty-Convex-Set)} Let $P$ be a set of points in $\R^d$ and $k \in \N$. Is there a set $Q \subseteq P$ of $k$ points in convex position whose convex hull does not contain any other point from $P$?
\end{definition}

Using the reduction technique from Giannopoulos et al. \cite{GKWW11arxiv}, it is a mere exercise to show that both problems are NP-hard if the dimension is not fixed. For people familiar with parameterized complexity: the problem is even W[1]-hard with respect to the dimension $d$. This means that it is very unlikely to admit an algorithm with running time $O\left(f(d) n^c\right)$ for \emph{any} computable function $f$ and constant $c$.

Still, this does not exclude the possibility that in every fixed dimension, the problem can be solved with a running time of, say, $O(n^{d+1})$. In this paper, we show that this cannot be the case (under standard complexity theoretic assumptions):
\begin{theorem}\label{Thm:Main} The problems \textsc{Largest-Empty-Convex-Set} and \textsc{Erd\H os-Szekeres} are NP-hard in $\R^3$.
\end{theorem}

The first part of the theorem, hardness of \textsc{LECS}, is shown in Sec. \ref{Sec:Reduction}. In Sec. \ref{Sec:Adaption}, the proof is adapted to \textsc{ES}. In Sec. \ref{Sec:WeakEps}, we derive a similar result for testing weak $\eps$-nets and red-blue discrepancy. Finally, in Sec. \ref{Sec:Conclusion}, we make several suggestions for further research on the subject.

\section{The reduction}\label{Sec:Reduction}
We will show that the problems is NP-hard by a reduction from the following problem:
\begin{definition} \textsc{(Independent-Set-for-Nonoverlapping-Unit-Disks)} Given a set of pairwise non-overlapping unit disks in $\R^2$, decide whether there are $k$ disks such that no two of them touch.
\end{definition}
Here, non-overlapping means that the interiors of the disks are pairwise disjoint. As shown by Cerioli et al. \cite{CFFP04}, the problem \textsc{ISNUD} is NP-hard. We will now reduce this problem to \textsc{LECS} and show how to adapt it to \textsc{ES} in the next section.

For a given instance $\cal D$ of unit disks in the plane, we will create a set of points in $\R^3$. These points will \emph{almost} lie on the elliptic paraboloid, in a sense to be made precise later.

For a point $x = (x_1, x_2) \in \R^2$, let \[ \lift\colon (x_1, x_2) \to (x_1, x_2, x_1^2 + x_2^2) \] denote the standard lifting transform to the paraboloid.

Let $D_c$ denote the $n$ centers of the disks in $\cal D$. Let $L$ denote the set of all points $\hat x = \lift (x)$, for $x \in D_c$.

We now want to forbid certain pairs of points to lie in empty convex positions, namely those for which the corresponding disks intersect. Thus, for a pair of intersecting disks $d, d'$ and their centers $c_d, c_{d'}$, we add a blocking point \[ b_{d d'} = \frac{1}{2} \left( \lift (c_{d'}) + \lift(c_d) \right). \] The set $B$ then consist of all the points $\{ b_{dd'} \mid d \cap d' \neq \emptyset \}$, and we set $P = L \uplus B$.

Thus, we have created $O(|\cal D|)$ points and the reduction is linear in the input size. The main property of the reduction is captured in the folowing lemma.

\begin{proposition}\label{Proposition:Separation} Let $Q$ be a set of points and $h$ be a hyperplane such that $hx \geq 0$ for all $x \in Q$. Then a point $x$ is in $\ch(Q)$ if and only if it is in $\ch(Q \cap h)$.
\end{proposition}

\begin{lemma}\label{Lemma:Encoding} A blocking point $b_{dd'}$ is contained in the convex hull of a set $Q \subseteq L$ if and only if $\hat c_d$ and $\hat c_{d'}$ are contained in $Q$.
\end{lemma}
\begin{proof} $\Leftarrow$: by definition\\
$\Rightarrow$: We show that there is a hyperplane that contains $b_{dd'}$, $\hat c_d$, and $\hat c_{d'}$ and has all other points strictly on the positive side. Here we will make use of the fact that our instance consists of non-overlapping unit disks - otherwise, the claim would not hold.

Let $C$ be the circle with center $\proj(b_{dd'})$ through $c_d$ and $c_{d'}$. Because all disks are non-overlapping unit disks, this circle does not contain any other points from $D_c$. We then take as $h$ the unique hyperplane whose intersection with the paraboloid projects to the circle $C$. This hyperplane contains all three points, and because the $C$ did not contain any points from $D_c$, all other points from $L$ and thus $B$ lie strictly above $h$. The claim then follows from Proposition \ref{Proposition:Separation}.
\end{proof}
The following states that whether or not a set is in empty convex position will depend only on which points we choose from $L$. The set $B$ can always be added without destroying this property.

\begin{proposition}\label{Observation:Convex} The sets $L$ and $B$ each are in empty convex position, and $\conv(L) = \conv(L \cup B)$. 
\end{proposition}
\begin{proof} By construction, all points of $L$ lie on the paraboloid. The points from $B$ can be separated from each other by the hyperplane defined in the previous proof. As all of them are convex combinations of points in $L$, we have $\ch(B) \subseteq \ch(L)$.
\end{proof}

\begin{corollary} A set $S = L' \uplus Q' \subseteq P$ is in convex position if and only if no point of $Q' \subseteq Q$ is contained in the convex hull of $L' \subseteq L$.
\end{corollary}

The main lemma then reads as follows:
\begin{lemma}\label{Lemma:Main} There is an independent set of size $m$ among the unit disks if and only if there are $m + |B|$ points in empty convex position.
\end{lemma}
\begin{proof}
$\Rightarrow$: Let $I$, $|I| = m$, be an independent set among the set of disks. Let $\hat I \subset L$ denote the corresponding lifted centers. We claim that $S = \hat I \cup B$ is in empty convex position. Indeed, by Observation \ref{Observation:Convex}, no point of $L - \hat I$ is in the convex hull of $S$. Further, by Lemma \ref{Lemma:Encoding}, if some point $b \in B$ was in $\ch(S)$, this would mean that there are two points in $\hat I$ that contained $b$. Thus, the corresponding disks would touch, and $I$ would not be an independent set. This means that there are $m + |B|$ points in empty convex position.\\
$\Leftarrow$: Now assume that there is no independent set of size $m$. This means that for any choice of $m$ disks, two of them touch. Now take any set $S$ of $m + |B|$ points. As there are only $|L| + |B|$ points in total, this must contain at least $m$ points from $L$. Thus, some two of them belong to disks that intersect. By Lemma \ref{Lemma:Encoding}, their convex hull contains a point of $B$. Thus, $S$ is not in empty convex position.
\end{proof}

\section{Adaption to Erd\H os-Szekeres}\label{Sec:Adaption}
We now show how this reduction can be applied to \textsc{Erd\H os-Szekeres}. One direction of Lemma \ref{Lemma:Main} is clear, since we have shown how an independent set of size $m$ results in an empty convex set of size $m + |B|$. For the other direction, we need to show that if there is \emph{any} not necessarily empty convex set of $m + |B|$ points, then there is also an independent set of size $m$ among the disks.

\begin{lemma} There is an independent set of size $m$ among the unit disks if and only if there are $m + |B|$ points in convex position.
\end{lemma}
\begin{proof} $\Rightarrow$: An empty convex set is convex.\\
$\Leftarrow$: Let $S$  be a set of $m + |B|$ points in convex position with $|S \cap B| < |B|$. We show how to construct a set $S'$ in convex position of the same size such that $|S' \cap B| = |S \cap B| + 1$. 

Let $I = S \cap L$, and let $D_I$ be the corresponding set of disks. Observe that, if $|S \cap B| < |B|$, then $|I| > m$. If all disks from $D_I$ are independent, we are done. Otherwise, let $d$ and $d'$ be two disks from $D_I$ that intersect. The point $b_{dd'}$ cannot be part of $S$, for otherwise $S$ would not be in convex position. If we thus set $S' = I - \{ \hat d \} \cup B \cup \{ b_{dd'} \}$, the set is still in convex position and we have $|S'| = |S|$ and $|S' \cap B| > |S\cap B|$. Thus, after finitely many steps we end up with a set of $m + |B|$ points which contains all points from $B$. In particular, the set contains no point from $B$ in the convex hull. But this means that the disks corresponding to these $m$ points from $L$ do not intersect. Thus, we have an independent set of size $m$.
\end{proof}

This finishes the proof of Thm. \ref{Thm:Main}.

\section{Testing weak $\eps$-nets and red-blue discrepancy}\label{Sec:WeakEps}
Here we shortly mention that the hardness proofs also show hardness for two closely related problems. Recall that a range space is a pair $(X, \mathcal R)$, where $\mathcal R \subset 2^X$. If $X$ is a set of points in $\R^d$ and $\mathcal R$ is the set of all convex sets determined by them, a weak $\eps$-net for $(X, \mathcal R)$ is a set of points $S$ such that $|S \cap R| \neq \emptyset$ whenever $|R \cap X| \geq \eps |X|$, for all $R \in \mathcal R$. We then define the corresponding decision problem as follows:
\begin{definition} \textsc{($\eps$-Net-Verification)} Given a set of points $P \subset \R^d$, another set $S \subset \R^d$ and an $\eps > 0$. Is $S$ an $\eps$-net for $P$ with respect to all convex sets?
\end{definition}
Chazelle et al. \cite{CEEGGSW95} give a polynomial time algorithm for the problem in the plane and ask whether it is solvable in polynomial time in $\R^3$.

A closely related concept is that of red-blue discrepancy: Given a set $R$ of red and a set $B$ of blue points, the \emph{discrepancy} of a set $C$ is defined as $D(C) = \left| |R \cap C| - |B \cap C| \right|$. The discrepancy of the set $P = R \cup B$ is then defined as $D(P) = \max_{C} D(C)$. The corresponding decision problem \textsc{Red-Blue-Discrepancy} asks whether the discrepancy of a given set is at least some value $k \in \N$.

Now observe that the set of blocking points $B$ determines an $(m/n)$-net\footnote{Recall that $n$ was the number of unit disks.} for the set of lifted points $L$ if and only if there is \emph{no} independent set of size $m$ among the disks. A similar argument holds for \textsc{Red-Blue-Discrepancy}. Our proof then also shows the following:
\begin{theorem} The problem $\eps$-\textsc{Net-Verification} is co-NP-complete in $\R^3$ and \textsc{Red-Blue-Discrepancy} is NP-hard in $\R^3$.
\end{theorem}

\section{Conclusion and open problems}\label{Sec:Conclusion}
This is work in progress (even though very little progress has been made in the past few weeks). In the future, we will try to extend the paper in the following direction.

The major open question is how to find an approximation algorithm for the problems \textsc{Erd\H os-Szekeres} and \textsc{Largest-Empty-Convex-Set}. The obvious approach (projecting to $\R^2$ and solving the problem there) does not work very well: as shown by Chazelle et al. \cite{CEG89}, there are polytopes whose projection in any direction has $\Theta(\log n)$ vertices on the convex hull. While this leads to a polynomial time $(\log n)/n$-approximation, only very few people will find this satisfying. Thus, the question for a more intelligent (probably constant-factor) approximation algorithm remains and seems to be very challenging.

In addition to this, the most interesting question is maybe the following: Is \textsc{Largest-Empty-Convex-Set} in $\R^3$ fixed parameter tractable with respect to the size of the solution? That is, can we decide whether there are $k$ points in empty convex position in time $O\left(f(k) \cdot n^c\right)$ for some computable function $f$ and constant $c$? More generally, given a point set $P$ in $\R^d$, can we decide whether there is an empty convex set of size $k$ in time $O\left(f(k)n^{O(d)}\right)$? 

Observe that due to the Erd\H os-Szekeres theorem itself, the problem \textsc{Erd\H os-Szekeres} is trivially fixed-parameter tractable: Given a point set $P$ and a $k \in \N$, if $n := |P| \leq 2^k$, we use a brute force algorithm, i.e., simply try all subsets of size $k$. This takes time ${n \choose k} \approx n^k \leq (2^k)^k$. If $n > 2^k$, we simply answer yes. In any case, the running time is bounded by $2^{k^2}$, and thus we have an algorithm with running time $O\left(f(k)n\right)$. Still, the question for a polynomial size problem kernel remains.

\paragraph*{Acknowledgements} We thank G\" unter Rote and Nabil Mustafa for pointing us out to \cite{CEEGGSW95} and \cite{CEG89}, respectively.

\bibliography{EmptyConvexSet}

\begin{thebibliography}{10}

\bibitem{AR85}
D.~Avis and D.~Rappaport.
\newblock Computing the largest empty convex subset of a set of points.
\newblock In {\em Proceedings of the first annual Symposium on Computational
  geometry}, SCG '85, New York, NY, USA, 1985. ACM.

\bibitem{BK01}
I.~B\'ar\'any and G.~K\'arolyi.
\newblock {Problems and Results around the Erd\"os-Szekeres Convex Polygon
  Theorem}.
\newblock In {\em Discrete and Computational Geometry}, volume 2098 of {\em
  Lecture Notes in Computer Science}. Springer, 2001.

\bibitem{CFFP04}
M.~Cerioli, L.~Faria, T.~Ferreira, and F.~Protti.
\newblock On minimum clique partition and maximum independent set on unit disk
  graphs and penny graphs: complexity and approximation.
\newblock {\em Electronic Notes in Discrete Mathematics}, 18, 2004.

\bibitem{CEEGGSW95}
B.~Chazelle, H.~Edelsbrunner, D.~Eppstein, M.~Grigni, L.~Guibas, M.~Sharir, and
  E.~Welzl.
\newblock Algorithms for weak epsilon-nets, 1995.

\bibitem{CEG89}
B.~Chazelle, H.~Edelsbrunner, and L.~J. Guibas.
\newblock The complexity of cutting complexes.
\newblock {\em Discrete {\&} Computational Geometry}, 4, 1989.

\bibitem{CK80}
V.~Chv\'atal and G.~Klincsek.
\newblock Finding largest convex subsets.
\newblock {\em Congresus Numeratium 29}, 1980.

\bibitem{DEO90}
D.~Dobkin, H.~Edelsbrunner, and M.~Overmars.
\newblock Searching for empty convex polygons.
\newblock {\em Algorithmica 5}, 1990.

\bibitem{ES35}
P.~Erd\"os and G.~Szekeres.
\newblock A combinatorial problem in geometry.
\newblock {\em Compositio Math. 2}, 1935.

\bibitem{Ger08}
T.~Gerken.
\newblock Empty convex hexagons in planar point sets.
\newblock {\em Discrete Comput. Geom.}, 39, 2008.

\bibitem{GKWW11arxiv}
P.~Giannopoulos, C.~Knauer, M.~Wahlstr{\"o}m, and D.~Werner.
\newblock Hardness of discrepancy computation and epsilon-net verification in
  high dimension.
\newblock {\em CoRR}, abs/1103.4503, 2011.

\bibitem{Ho83}
J.~D. Horton.
\newblock Sets with no empty convex 7-gons.
\newblock {\em C. Math. Bull. 26}, 1983.

\bibitem{MRSW95}
J.~S.~B. Mitchell, G.~Rote, G.~Sundaram, and G.~Woeginger.
\newblock Counting convex polygons in planar point sets.
\newblock {\em Information Processing Letters 56}, 1995.

\bibitem{MS00}
W.~Morris and V.~Soltan.
\newblock The {Erd\"os-Szekeres} problem on points in convex position -- a
  survey.
\newblock {\em Bull. Amer. Math. Soc. 37}, 2000.

\bibitem{Nic07}
C.~M. Nicol\'as.
\newblock The empty hexagon theorem.
\newblock {\em Discrete Comput. Geom.}, 38, September 2007.

\end{thebibliography}

\end{document}